\documentclass[a4paper,twoside,reqno]{amsart}
\usepackage{amsmath}
\usepackage{amsfonts}
\usepackage{amssymb}
\usepackage{amsthm}
\usepackage{newlfont}
\usepackage{graphicx}
\usepackage{amscd}
\usepackage{young}

\theoremstyle{plain}
\newtheorem{Th}{Theorem}[section]
\newtheorem{Cor}[Th]{Corollary}

\newtheorem{Prop}[Th]{Proposition}
\theoremstyle{definition}
\newtheorem{Def}{Definition}[section]
\newtheorem{Ex}{Example}[section]
\theoremstyle{remark}
\newtheorem*{Rem}{Remark}
\numberwithin{equation}{section}

\newcommand{\DD}{{\mathbb D}}

\newcommand{\ZZ}{{\mathbb Z}}
\newcommand{\VV}{{\mathbb V}}

\newcommand{\bpsi}{\boldsymbol{\psi}}

\begin{document}

\title
{Non-commutative Hermite--Pad\'{e} approximation\\ and integrability}

\author{Adam Doliwa}

\address{Faculty of Mathematics and Computer Science\\
University of Warmia and Mazury in Olsztyn\\
ul.~S{\l}oneczna~54\\ 10-710~Olsztyn\\ Poland} 
\email{doliwa@matman.uwm.edu.pl}

%
\date{}

\begin{abstract}
We introduce and solve the non-commutative version of the Hermite--Pad\'{e} type I approximation problem. Its solution, expressed by quasideterminants, leads in a natural way to a subclass of solutions of the non-commutative Hirota (discrete Kadomtsev--Petviashvili) system and of its linear problem. We also prove integrability of the constrained system, which in the simplest case is the non-commutative discrete-time Toda lattice equation known from the theory of non-commutative Pad\'{e} approximants and matrix orthogonal polynomials. 
\end{abstract}
\maketitle

\section{Introduction}
Hermite--Pad\'{e} approximation technique, originally introduced by Hermite to prove transcendency of Euler's constant~\cite{Hermite,Hermite-P}, has attracted recently considerable attention in mathematical physics due to close connection to multiply orthogonal polynomials, random matrices, diffusion models or various combinatorial problems~\cite{Aptekarev,BakerGraves-Morris,Kuijlaars,BleherKuijlaars,AptekarevKuijlaars,VanAsche}. An important ingredient behind the success of these theories was the integrability of resulting equations~\cite{AdlervanMoerbeke,AptekarevDerevyaginVanAssche,Alvarez-FernandezPrietoManas,AFACGAMM,Filipuk-VanAssche-Zhang,LopezLagomasinoMedinaPeraltaSzmigielski,ManoTsuda,Miranian,SinapVanAssche}. Recently a direct link between the Hermite--Pad\'{e} approximants and integrability has been found~\cite{Doliwa-Siemaszko-HP} in terms of the corresponding reduction of Hirota's discrete Kadomtsev--Peviashvili (KP) system, which in the simplest case of the Pad\'{e} rational approximation gives the discrete-time Toda lattice equations. Actually, it turned out that the relevant difference equations were known in the numerical algorithms community~\cite{Paszkowski,DD-DC-2} including their special solutions in terms of certain determinants.

Hirota's discrete KP system plays a special role within the theory of integrable equations and their applications. As it was shown by Miwa~\cite{Miwa}, it encodes the full KP hierarchy~\cite{DKJM} of integrable partial differential equations. It is well known, see for example reviews~\cite{KNS-rev,Zabrodin}, that majority of the known integrable systems can be obtained as its reductions. Moreover, the most important techniques used to find solutions of integrable equations can be applied to the Hirota system in their pure forms (the finite gap algebro-geometric method~\cite{BBEIM} in~\cite{Krichever,Shiota}, the  non-local $\bar{\partial}$-dressing method~\cite{AblowitzBarYaacovFokas,Konopelchenko-book} in~\cite{Dol-Des}, the Darboux transform~\cite{Matveev} in~\cite{Nimmo-KP}). It has simple geometric meaning in terms of geometric configurations~\cite{Dol-Des}, and its multidimensional consistency is encoded in the (fundamental to projective geometry) Desargues configuration. The symmetry structure of the equations is described by affine Weyl groups of A-type~\cite{Dol-AN}.

In the present paper we transfer the basic elements of the connection between the Hermite--Pad\'{e} approximants and integrability to the non-commutative level. Searching for non-commutative ge\-ne\-ra\-li\-za\-tions is well motivated by physics, but also by theoretical computer science and combinatorics. 
An important technical tool used in in such transition is provided by quasideterminants~\cite{Quasideterminants-GR1,Quasideterminants,KrobLeclerc}, which in a sense replace the standard determinants (better to say their ratios) in the non-commutative linear algebra. In fact, the fundamental properties of the non-commutative Pad\'{e} approximants~\cite{Draux-rev}, including also their connection to non-commutative/matrix orthogonal polynomials~\cite{Draux-OP-PA}, can be formulated in terms of quasideterminants~\cite{NCSF} see also~\cite{Miranian,SinapVanAssche,AFACGAMM,Shi-HaoLi,Doliwa-Siemaszko-W}. We follow this idea in replacing by quasideterminants the ratios of determinants in the corresponding formulas of~\cite{Doliwa-Siemaszko-HP}, where the link between the Hermite--Pad\'{e} approximants and integrability has been established in the commutative case.

The fully non-commutative Hirota system (originally called the non-Abelian Hirota--Miwa system) was proposed in~\cite{Nimmo-NCKP} where the corresponding Darboux transform was given as well in terms of quasideterminants, see also~\cite{LiNimmoTamizhmani}. 	
The class of solutions proposed below is of different nature. The reduction of the non-commutative Hirota system given in the present work differs also substantially from the known periodic reductions~\cite{Dol-GD,DoliwaNoumi}. 
Other application of the quasideterminants to non-commutative discrete integrable systems can be found in~\cite{DiFrancescoKedem,Doliwa-NCCF}.

The structure of paper is as follows. In Section~\ref{sec:preliminaries} we first present the relevant information on the non-commutative Hirota system, and then we recall basic properties of quasideterminants. Section~\ref{sec:NCHP} is devoted to formulation of the non-commutative Hermite--Pad\'{e} approximation problem and to presentation of its solution in terms of quasideterminants. We also establish there a link with the non-commutative Hirota system and its linear problem. Then in Section~\ref{sec:NCDT} we show that the relevant solutions of the system satisfy an additional integrable constraint. We also give arguments why the reduced system can be called non-commutative multidimensional Toda system. At the end of the paper we summarize its results, present open questions and  future research directions. 

\section{Preliminaries} \label{sec:preliminaries}
\subsection{The non-Abelian Hirota--Miwa system}
Consider the following linear problem \cite{DJM-II,Nimmo-NCKP}
\begin{equation} \label{eq:lin-d-adj-KP}
\bpsi(n-e_i) - \bpsi(n-e_j) =  \bpsi(n) U_{ij}(n),  \qquad 1\leq i \ne j \leq N,
\end{equation}
where $U_{ij}\colon\ZZ^N \to \DD$ are functions defined on $N$-dimensional integer lattice, $N\geq 3$, with values in a division ring $\DD$, and the wave function $\bpsi\colon \ZZ^N \to \VV(\DD)$ takes values in a right vector space over $\DD$; here  $n = (n_1, \dots , n_N) = \sum_{i=1}^N n_i e_i \in \ZZ^N$, and $e_i$ is the element of the standard basis in $\ZZ^N$ lattice.

The compatibility conditions of \eqref{eq:lin-d-adj-KP} consist of equations 
\begin{equation} \label{eq:alg-comp-U}
\begin{split}
U_{ij}(n) + U_{ji}(n) = 0, \quad  U_{ij}(n) & + U_{jk}(n) + U_{ki}(n) = 0,  \\ 
U_{ij}(n)U_{ik}(n-e_j) = & U_{ik}(n) U_{ij}(n-e_k),
\end{split} \qquad i,j,k \quad \text{distinct},
\end{equation}
called in \cite{Nimmo-NCKP} non-Abelian Hirota--Miwa system. 
From the last part of compatibility conditions one can deduce existence of the potentials $\rho^j\colon \ZZ^N \to \DD$ such that
\begin{equation} \label{eq:U-rho}
U_{ij}(n) = - \rho^i (n) \left[ \rho^i(n-e_j) \right]^{-1} ,
\end{equation}
which are given up to arbitrary functions of single variables.
Apart from the constraint 
\begin{equation} \label{eq:rho-rho}
\rho^j (n) \left[ \rho^j(n-e_i) \right]^{-1} + \rho^i (n) \left[ \rho^i(n-e_j) \right]^{-1} =0 ,
\end{equation}
which is the first part out of the system \eqref{eq:alg-comp-U}, the second part gives 
\begin{equation} \label{eq:rho-rho-rho}
\rho^i (n) \left[ \rho^i(n-e_j) \right]^{-1} + \rho^j (n) \left[ \rho^j(n-e_k) \right]^{-1} + \rho^k (n) \left[ \rho^k(n-e_i) \right]^{-1} =0.
\end{equation}
\begin{Rem}
For commutative $\DD$ it is possible to introduce the single potential function $\tau(n)$ such that
\begin{equation} \label{eq:U-tau}
U_{ij}(n) = \frac{\tau(n) \tau(n-e_i-e_j)}{\tau(n-e_i) \tau(n-e_j)}, \qquad i< j.
\end{equation}
Then the remaining (second) part of the system
\eqref{eq:alg-comp-U} reduces to Hirota's discrete KP equation \cite{Hirota,Miwa}
\begin{equation} \label{eq:H-M}
\tau(n-e_i)\tau(n-e_j-e_k) - \tau(n-e_j)\tau(n-e_i-e_k) + \tau(n-e_k)\tau(n-e_i-e_j) =0,
\end{equation}
where $1\leq i< j <k \leq N$. 
\end{Rem}
\begin{Rem}
In order to adjust to results of Section~\ref{sec:NCHP}, instead of the original formulation of the non-Abelian Hirota--Miwa system and of its linear problem  we use their dual (called also adjoint) versions~\cite{GilsonNimmoOhta,LiNimmoTamizhmani}. The original linear problem~\cite{Nimmo-NCKP} is obtained by replacing the shifts into negative directions by positive direction shifts.
\end{Rem}

\subsection{Quasideterminants}

In this Section we recall, following~\cite{Quasideterminants-GR1},  the definition and basic properties of quasideterminants. 
\begin{Def}
	Given square matrix $X=(x_{ij})_{i,j=1,\dots,n}$ with formal entries $x_{ij}$. In the free division ring~\cite{Cohn} generated by the set $\{ x_{ij}\}_{i,j=1,\dots,n}$ consider the formal inverse matrix $Y=X^{-1}= (y_{ij})_{i,j=1,\dots,n}$ to $X$.
	The $(i,j)$th quasideterminant $|X|_{ij}$ of $X$ is the inverse $(y_{ji})^{-1}$ of the $(j,i)$th element of $Y$, and is often written explicitly as
	\begin{equation}
	|X|_{ij} = \left| \begin{matrix}
	x_{11} & \cdots & x_{1j} & \cdots & x_{1n} \\
	\vdots &        & \vdots &        & \vdots \\
	x_{i1} & \cdots & \boxed{x_{ij}} & \cdots & x_{in}  \\
	\vdots &        & \vdots &        & \vdots \\
	x_{n1} & \cdots & x_{nj} & \cdots & x_{nn} 
	\end{matrix} \right| .
	\end{equation}	
\end{Def}
Quasideterminants can be computed using the following recurrence relation. For $n\geq 2$ let $X^{ij}$ be the square matrix obtained from $X$ by deleting the $i$th row and the $j$th column (with index $i/j$ skipped from the row/column enumeration), then
\begin{equation} \label{eq:QD-exp}
|X|_{ij} = 
x_{ij} - \sum_{\substack{ i^\prime \neq i \\ j^\prime \neq j }} x_{i j^\prime} |X^{ij}|_{i^\prime j^\prime }^{-1} x_{i^\prime j}
\end{equation}
provided all terms in the right-hand side are defined.
\begin{Rem}
	When the elements of the matrix $X$ commute between themselves, what we denote by placing the letter $c$ over the equality sign, then the familiar matrix inversion formula gives
\begin{equation} \label{eq:qdet-det}
|X|_{ij} \stackrel{c}{=} (-1)^{i+j}\frac{\det X}{\det X^{ij}}.
\end{equation}
\end{Rem}
\begin{Ex}
Quasideterminants of generic $2\times 2$ matrix
\begin{equation*}
X = \left( \begin{matrix}
x_{11} & x_{12} \\
x_{21} & x_{22}
\end{matrix}\right)
\end{equation*}	
read as follows
\begin{gather*}
|X|_{11} = x_{11} - x_{12} x_{22}^{-1} x_{21}, \quad 
|X|_{12} = x_{12} - x_{11} x_{21}^{-1} x_{22}, \\ 
|X|_{21} = x_{21} - x_{22} x_{12}^{-1} x_{11}, \quad 
|X|_{22} = x_{22} - x_{21} x_{11}^{-1} x_{12}.
\end{gather*}
\end{Ex}

Let us collect basic properties of the quasideterminants which will be used in the sequel; see also~\cite{KrobLeclerc}.
\subsubsection{Row and column operations}
\begin{itemize}
	\item 
	The quasideterminant $|X|_{ij}$ does not depend on permutations of rows and columns in the matrix $X$ that do not involve the $i$th row and the $j$th column.	
	\item Let the matrix $\tilde{X}$ be obtained from the matrix $X$ by multiplying the $k$th row by the element $\lambda$ of the division ring from the left, then 
	\begin{equation}
	|\tilde{X}|_{ij} = \begin{cases} \lambda|X|_{ij} & \text{if} \quad i = k, \\
	|X|_{ij} & \text{if} \quad i\neq k \quad \text{and} \; \lambda \; \text{is invertible} . \end{cases}
	\end{equation}
	
	\item Let the matrix $\hat{X}$ be obtained from the matrix $X$ by multiplying the $k$th column by the element $\mu$ of the division ring from the right, then 
	\begin{equation}
	|\hat{X}|_{ij} = \begin{cases} |X|_{ik} \, \mu & \text{if} \quad j = k, \\
	|X|_{ij} & \text{if} \quad j\neq k \quad \text{and} \; \mu \; \text{is invertible} . \end{cases}
	\end{equation}
	
	\item
	Let the matrix $\tilde{X}$ be constructed by adding to some row of the matrix $X$ its $k$th row multiplied by a scalar $\lambda$ from the left, then
	\begin{equation}
	|X|_{ij} = |\tilde{X}|_{ij}, \qquad i = 1, \dots , k-1, k+1, \dots , n, \quad j=1,\dots , n.
	\end{equation}
	
\item Let the matrix $\hat{X}$ be constructed by addition to some column of the matrix $X$ its $l$th column multiplied by a scalar $\mu$ from the right, then
	\begin{equation}
	|X|_{ij} = |\hat{X}|_{ij}, \qquad i = 1, \dots , n, \quad j=1,\dots , l-1, l+1 , \dots ,n.
	\end{equation}
	
\end{itemize}
\subsubsection{Homological relations}
\begin{itemize}
	\item Row homological relations:
	\begin{equation} \label{eq:row-hom}
	-|X|_{ij} \cdot |X^{i k}|_{sj}^{-1} = |X|_{ik} \cdot |X^{ij}|_{sk}^{-1}, \qquad s\neq i.
	\end{equation}
	\item Column homological relations:
	\begin{equation} \label{eq:chr}
	- |X^{kj}|_{is}^{-1} \cdot |X|_{ij}=  |X^{ij}|_{ks}^{-1} \cdot |X|_{kj} , \qquad s\neq j.
	\end{equation}
\end{itemize}

\subsubsection{Sylvester's identity}

Let $X_0 = (x_{ij})$, $i,j = 1,\dots ,k$, be a submatrix of $X$ that is invertible. For $p,q = k+1,\dots ,n$ set
\begin{equation*}
c_{pq} = \begin{vmatrix}
&&& x_{1q} \\
& X_0 & & \vdots\\
&&& x_{kq} \\
x_{p1} & \dots & x_{pk} & \boxed{x_{pq}} 
\end{vmatrix} \; ,
\end{equation*}
and consider the $(n-k) \times (n-k)$ matrix $C = (c_{pq})$, $p,q = k+1,\dots , n$. Then for $i,j = k+1,\dots , n$,
\begin{equation}
|X|_{ij} = |C|_{ij} \; .
\end{equation}
In applications  the Sylvester identity is usually used in conjunction with row/column permutations.

\section{Non-commutative Hermite--Pad\'{e} approximants} \label{sec:NCHP}
\subsection{Formulation of the problem}
Consider $m$ formal series $\left( f_1(x),\dots , f_m(x)\right) $ in variable $x$ with non-commuting coefficients  
\begin{equation}
	f_i(x) = \sum_{j=0}^\infty f^i_j x^j, 
\end{equation}
where the parameter $x$ commutes with all the coefficients.
Given $n=(n_1,\dots , n_m)$ element of $\ZZ_{\geq -1}^m$, we write also $|n|= n_1 + \dots + n_m$. A \emph{Hermite--Pad\'{e} form of degree $n$} is, by definition, every system of polynomials $(Y_1(x), \dots , Y_m(x))$, not all equal to zero, with corresponding degrees $\deg Y_i(x) \leq n_i$, $i=1,\dots , m$ (degree of the zero polynomial equals $-1$), and such that 
\begin{equation} \label{eq:HP-cond}
f_1(x) Y_1(x) + \dots +  f_m(x) Y_m(x) = x^{|n|+m-1}\Gamma(x)
\end{equation}
for a series $\Gamma(x) = \sum_{j=0}^\infty \Gamma^j x^j $.
\begin{Rem}
	In our paper we consider only the so called type I approximants. For closely reated type II and mixed type approximants in the commutative case see, for example \cite{Mahler-P}. 
\end{Rem}

\subsection{Solution of the problem}
The degree conditions and equation~\eqref{eq:HP-cond} lead to a system of $|n|+m-1$ linear equations with $|n|+m$ unknown coefficients of the polynomials.
Define matrix $\mathcal{M}(n)$ of $(|n|+m-1)$ rows and $(|n|+m)$ columns 
\begin{equation}
\mathcal{M}(n) = \left( \begin{array}{cccccccc}
f^1_0 &     & 0 & & & f^m_0 &   & 0 \\
f^1_1 & \ddots    &   & \cdots & \cdots & f^m_1 & \ddots  &   \\
\vdots &   &  f^1_0 & & & \vdots &    &  f^m_0 \\
\vdots &   &  \vdots & \cdots &\cdots & \vdots &  & \vdots\\
f^1_{|n|+m-2} & \cdots & f^1_{|n|+m - n_1 -2} & & & f^m_{|n|+m-2} & \cdots & f^m_{|n|+m - n_m -2}
\end{array} \right) ;
\end{equation}
its columns are divided into $m$ groups, the $i$th group is composed out of $n_i +1$ columns depending on $f_i(x)$ only. Supplement the matrix $\mathcal{M}(n)$ by the line
\begin{equation} \label{eq:bN}
\left( f^1_{|n|+m-1}, f^1_{|n|+m-2}, \dots , f^1_{|n|+m-n_1 -1}, \cdots , \cdots , f^m_{|n|+m-1},  \dots , f^m_{|n|+m-n_m-1} \right),
\end{equation} 
as the last row, and define $\rho^i(n)$ as the quasideterminant of such extended matrix, with respect to the element in the last row and the last column of the $i$th block.

\begin{Prop}
	By the row homological relations \eqref{eq:row-hom} the functions $\rho^j$, $j=1,\dots , m$, when exist, satisfy equations 
\begin{equation*} 
\label{eq:rho-hr}
\rho^i(n) [ \rho^i(n-e_j)]^{-1} = - \rho^j(n) [ \rho^j(n-e_i)]^{-1}, \qquad i\neq j. 
\end{equation*}
\end{Prop}
One can observe that the above system is identical with equations~\eqref{eq:rho-rho} of the theory of the non-Abelian Hirota--Miwa system; we keep therefore the same notation. Moreover we are going to show how other elements of the theory can be obtained from the data of the Hermite--Pad\'{e} approximants.

Supplement the matrix $\mathcal{M}(n)$ at the bottom by the line 
\begin{equation}
\label{eq:Xk}
X_k = (0,\dots, 0, \dots \; \dots , 1, x , \dots , x^{n_k}, \dots \; \dots , 0, \dots ,0),
\end{equation}
consisting of zeros except for the $k$th block of the form $1, x, \dots , x^{n_k}$. Define $Z_k^{(i)}(n,x)$ as the quasideterminant of such matrix with respect to the element in the last row and the last column of the $i$th block. 

\begin{Prop} \label{prop:Z_k^i}
	When it exists, $Z_k^{(i)}(n,x)$ is a polynomial in $x$ of degree not greater than $n_k$. In particular, $Z_k^{(k)}(n,x)$ is monic.
\end{Prop}
\begin{proof}
	By the elementary row operations and the recurrence~\eqref{eq:QD-exp} one can decompose $Z_k^{(i)}(n,x)$ into the sum of quasideterminants multiplied by subsequent powers of the parameter $x$. The highest order term of $Z_k^{(k)}(x,n)$ can be easily calculated using equation~\eqref{eq:QD-exp}.
\end{proof}
\begin{Prop} \label{prop:Z-psi}
The product 
\begin{equation} \label{eq:psi-Z}
\psi_k(n,x) = Z_k^{(i)}(n,x) [ \rho^i(n)]^{-1} , \qquad i,k = 1, \dots , m,
\end{equation}
is independent of the index $i$. This allows to find the leading term of the polynomials
\begin{align}
Z_k^{(i)}(n,x)  = &
[\rho^k(n)]^{-1} \rho^i(n) x^{n_k} + \text{lower order terms}
, \qquad i,k=1, \dots , m, \\
\psi_k(n,x)  = &
[\rho^k(n)]^{-1} x^{n_k} + \text{lower order terms}
, \qquad k=1, \dots , m.
\end{align}
\end{Prop}
\begin{proof}	
By combining the row homological relations 
\begin{equation} 
\label{eq:Z-hr}
Z_k^{(i)}(n,x) [ \rho^i(n-e_j)]^{-1} = - Z_k^{(j)}(n,x) [ \rho^j(n-e_i)]^{-1}, \qquad i\neq j,k=1, \dots , m,
\end{equation}
with analogous relations \eqref{eq:rho-hr} we obtain the first part of the statement. The second part follows from Proposition~\ref{prop:Z_k^i}.
\end{proof}

Let us present the role of the polynomials $\psi_k(n,x)$ within the non-commutative Hermite--Pad\'{e} approximation theory.
\begin{Prop}
	The polynomials $(\psi_{1}(n,x), \dots , \psi_m(n,x))$  provide solution of the non-commutative Hermite--Pad\'{e} problem with the following asymptotic
	\begin{equation} \label{eq:HP-psi}
	f_1(x) \psi_1(n,x) + \dots + f_m(x) \psi_m(n,x) = x^{|n| + m - 1} + \text{higher order terms}.
	\end{equation}
\end{Prop}
\begin{proof}
	Supplement the matrix $\mathcal{M}(n)$ at the bottom by the line 
	\begin{equation*}
	\left( f_1(x), xf_1(x), \dots , x^{n_1} f_1(x), \cdots \: \cdots , f_m(x), x f_m(x), \dots , x^{n_m} f_m(x) \right),
	\end{equation*}
	and calculate the quasi-determinant $QD_i$ of the resulting square matrix	with respect to the element in the last row and in the last column of the $i$th block. The direct calculation, by decomposition of the last row into blocks in equation~\eqref{eq:QD-exp}, gives
	\begin{equation*}
	QD_i = f_1(t) Z_1^{(i)}(n,t) + \dots + f_m(t) Z_m^{(i)}(n,t).
	\end{equation*}
	From the other hand, the row operations allow to remove from the last row terms of the order lower than $|n| + m - 1$ in the parameter $x$, what implies 
	\begin{equation*}
	QD_i = x^{|n| + m - 1} \rho^i(n) + \text{higher order terms}.
	\end{equation*}	
	Then equation~\eqref{eq:HP-psi} follows from definition~\eqref{eq:psi-Z}.
\end{proof}
\begin{Cor}
	Equivalently, $(Z_{1}^{(i)}(n,x),\dots , Z_m^{(i)}(n,x))$ where $i=1,\dots,m$, are also solutions of the non-commutative Hermite--Pad\'{e} problem with asymptotics
	\begin{equation} \label{eq:HP-Zi}
	f_1(x) Z_1^{(i)}(n,x) + \dots + f_m(x) Z_m^{(i)}(n,x) = 
	x^{|n| + m - 1} \rho^i(n) + \text{higher order terms}.
	\end{equation}	
\end{Cor}

\subsection{The non-commutative Hirota system}
Finally, we show the relation of the above solution of the non-commutative Hermite--Pad\'{e} problem with the non-commutative Hirota system.
\begin{Prop} The polynomials $\psi_k(n,x)$, $k=1,\dots , m$, satisfy the linear problem of the non-Abelian Hirota--Miwa system
	\begin{equation} \label{eq:psi-rho-HP}
	\psi_k(n-e_i,x) - 	\psi_k(n-e_j,x) = 	\psi_k(n,x) \rho^j(n) [ \rho^j(n-e_i)]^{-1}, \qquad i\neq j.
	\end{equation}
\end{Prop}
\begin{proof}
	Let us apply the Sylvester identity for matrix $\mathcal{M}(n)$ supplemented by $X_k$, with respect to last two rows and the last columns of the $i$th and $j$th blocks, where $i<j$. The corresponding $2\times 2$ matrix reads
	\begin{equation*}
	\left( \begin{matrix}
	\rho^i(n-e_j) & \rho^j(n-e_i) \\
	Z_{k}^{(i)}(n-e_j) & Z_{k}^{(j)}(n-e_i)
	\end{matrix} \right),
	\end{equation*} 
and the quasideterminant of the big matrix with respect to the element in the last row and the last column of the $j$th block equals then
	\begin{equation}
	Z_k^{(j)}(n,x) = Z_k^{(j)}(n-e_i,x) - 
	Z_k^{(i)}(n-e_j,x) [ \rho^i(n-e_j)]^{-1}  \rho^j(n-e_i). 
	\end{equation}
Multiplication of the above equation by $[\rho^j(n-e_i)^{-1}]$
from the right gives the statement for such an ordering of indices. The homological relations \eqref{eq:rho-hr} prove the statement for $j<i$.
\end{proof}

\begin{Rem}
	To avoid degenerations we usually assume that for all $n\in\ZZ^m_{\geq 0}$ the potentials $\rho^i(n)$ do not vanish. Such a system of series $(f_1, \dots , f_m)$ is called \emph{perfect}, in analogy to the commutative case \cite{Mahler-P}. Then, by properties of the quasideterminants~\cite{Quasideterminants} for each $n$ the Hermite--Pad\'{e} problem has a solution defined uniquely up to a constant factor with maximal degrees of the polynomials. In such a case there is an alternative way to derive the linear problem~\eqref{eq:psi-rho-HP}.
	The system of polynomials defined on the left hand side of \eqref{eq:psi-rho-HP} provides the solution of the Hermite--Pad'{e} problem for $n$, thus by the non-degeneracy condition must be proportional to the system $\psi_k(n,x)$. The coefficient is fixed by examining the highest order term of $k=j$th polynomial on both sides and Proposition~\ref{prop:Z-psi}. 
\end{Rem}
\begin{Cor}
	By examining the highest order term of $k$th polynomial, $k\neq i,j$, on both sides of the linear problem~\eqref{eq:psi-rho-HP} and using the homological relations~\eqref{eq:rho-hr} we obtain that the potentials $\rho^i$ satisfy the second part~\eqref{eq:rho-rho-rho} of the non-commutative Hirota equations.
\end{Cor}

\subsection{Comparison with the commutative case}
To close this Section let us show the link with the  commutative case~\cite{Doliwa-Siemaszko-HP}. By equation~\eqref{eq:qdet-det} we have
	\begin{equation}
	\rho^{i}(n) \stackrel{c}{=} (-1)^{n_{i+1} + \dots + n_m + m-i} \frac{\tau(n)}{\tau(n-e_i)},
	\end{equation}
	where $\tau(n)$ is the determinant of the square matrix used to define $\rho^i(n)$
\begin{equation*}
\tau(n) \stackrel{c}{=} \left| \begin{array}{cccccccc}
f^1_0 &     & 0 & & & f^m_0 &   & 0 \\
f^1_1 & \ddots    &   & \cdots & \cdots & f^m_1 & \ddots  &   \\
\vdots &   &  f^1_0 & & & \vdots &    &  f^m_0 \\
\vdots &   &  \vdots & \cdots &\cdots & \vdots &  & \vdots\\
f^1_{|n|+m-2} & \cdots &   & & &  & \cdots & f^m_{|n|+m - n_m -2} \\
f^1_{|n|+m-1} & \cdots & f^1_{|n|+m - n_1 -1} & & & f^m_{|n|+m-1} & \cdots & f^m_{|n|+m - n_m -1}
\end{array} \right| .
\end{equation*}	
Similarly
\begin{equation}
Z_k^{(i)}(n,x) \stackrel{c}{=} (-1)^{n_{i+1} + \dots + n_m + m-i} \frac{Z_k(n,x)}{\tau(n-e_i)},
\end{equation}
where $Z_k(n,x)$ is the determinant of the square matrix used to define $Z_k^{(i)}(n,x)$	
\begin{equation*} \label{eq:Z-det}
Z_k(n,x)  \stackrel{c}{=} \left|
\begin{smallmatrix}
f^1_0 &     & 0 & & & f^k_0 & & 0 && & f^m_0 &   & 0 \\
f^1_1 & \ddots    &   & \cdots & \cdots & f^k_1  &\ddots &  &\cdots &\cdots & f^m_1 & \ddots  &   \\
\vdots &   &  f^1_0 & & & \vdots && f^k_0 & && \vdots &     &  f^m_0 \\
\vdots &   &  \vdots & \cdots &\cdots &&&&\cdots &\cdots & \vdots &  & \vdots\\
f^1_{|n|+m-2} & \cdots & f^1_{|n|+m - n_1 -2} & & &f^k_{|n|+m-2}& \cdots & f^k_{|n|+m-n_k - 2}& & &  f^m_{|n|+m-2} & \cdots & f^m_{|n|+m - n_m -2} \\
0 & \cdots & 0 & \cdots & \cdots &1 & \cdots & x^{n_k} & \cdots &\cdots & 0 & \cdots & 0
\end{smallmatrix} \right|.
\end{equation*}
Such polynomials gives the so called canonical solution of the Hermite--Pad\'{e} problem
\begin{equation}
f_1(t) Z_1(n,t) + \dots + f_m(t) Z_m(n,t) = 
x^{|n| + m - 1} \tau(n) + \text{higher order terms}.
\end{equation}
Moreover, then 
\begin{equation}
\psi_k(n,x) \stackrel{c}{=} \frac{Z_k(n,x)}{\tau(n)}
\end{equation}
satisfies the standard adjoint linear problem 
of the Hirota discrete KP system~\eqref{eq:H-M}. Correspondingly, the polynomials $Z_k(n,x)$ satisfy the so called bilinear form of the adjoint linear problem
\begin{equation}
\label{eq:lin-Z}
Z_{k}(n,x) \tau(n-e_i-e_j) = Z_k(n-e_i,x) \tau(n-e_j)
	 - Z_k(n-e_j,x) \tau(n-e_i), \qquad 1\leq i<j \leq m.
\end{equation}

\section{The "multidimensional" non-commutative discrete-time Toda equation} \label{sec:NCDT}
\subsection{The non-commutative version of the Paszkowski constraint}
In the commutative case the polynomials $Z_k(n,x)$ in addition to equations~\eqref{eq:lin-Z} satisfy the constraint~\cite{Paszkowski}
\begin{equation}
\label{eq:P-constr}
xZ_k(n,x) \tau(n) = Z_{k}(n+e_1,x)\tau(n-e_1) + \cdots + Z_{k}(n+e_m,x)\tau(n-e_m), 
\end{equation}
which supplements Hirota's discrete KP system with the additional equation
\begin{equation}
\label{eq:P-eq}
[\tau(n)]^2 = \tau(n+e_1)\tau(n-e_1) + \cdots + \tau(n+e_m)\tau(n-e_m). 
\end{equation}
Our goal in this Section will be to give the non-commutative version of the above.

\begin{Th}
	Under the non-degeneracy assumption the polynomials $\bpsi(n,x) = (\psi_1(n,x),\dots , \psi_m(n,x))$, which provide solution of the non-commutative Hermite--Pad\'{e} problem, satisfy also the constraint
\begin{equation}
\label{eq:P-constr-nc}
x \bpsi(n,x) = \bpsi(n+e_1,x) \rho^1(n+e_1) [\rho^1(n)]^{-1} + \dots + \bpsi(n+e_m,x) \rho^m(n+e_m) [\rho^m(n)]^{-1} ,
\end{equation}
and the quasideterminants $\rho^j(n)$ satisfy the following equation
\begin{equation}
\label{eq:P-eq-nc-1}
1= \rho^1(n+e_1) [\rho^1(n)]^{-1} + \dots + \rho^m(n+e_m) [\rho^m(n)]^{-1} .
\end{equation}	
\end{Th}
\begin{proof}
The system of $m$ polynomials (the components of the vector)
\begin{equation*}
x \bpsi(n,x) - \bpsi(n+e_1,x) \rho^1(n+e_1) [\rho^1(n)]^{-1} - \dots + \bpsi(n+e_{m-1},x) \rho^{m-1}(n+e_{m-1}) [\rho^{m-1}(n)]^{-1} 
\end{equation*}	
forms a solution of the Hermite--Pad\'{e} problem with the same degrees as the system $\bpsi(n+e_m,x)$, thus both systems must be proportional. The coefficient can be found by examining the highest order term of the $m$th polynomial.
Equation \eqref{eq:P-eq-nc-1} follows from comparison of the leading order terms of both sides of~\eqref{eq:P-constr-nc} in the equation~\eqref{eq:HP-psi}.
\end{proof}
\begin{Rem}
	In the simplest case $m=2$ we are left with the system
	\begin{gather}
1= \rho^1(n+e_1) [\rho^1(n)]^{-1} + \rho^2(n+e_2) [\rho^2(n)]^{-1} ,\\
\rho^1(n) [ \rho^i(n-e_2)]^{-1} = - \rho^2(n) [ \rho^2(n-e_2)]^{-1}.
	\end{gather}
Elimination of $\rho^1(n)$ gives the non-commutative discrete Toda chain equation of the non-commutative Pad\'{e} theory~\cite{Doliwa-Siemaszko-W}
\begin{equation}
\rho^2(n+e_1) \left( [\rho^2(n-e_2)]^{-1} -  [\rho^2(n)]^{-1} \right) \rho^2(n-e_1) = \rho^2(n+e_2) - \rho^2(n).
\end{equation}	
A similar equation, connected to the above by the homological relations, was obtained recently within the theory of matrix orthogonal polynomials~\cite{Shi-HaoLi}. 
\end{Rem}
\begin{Rem}
In the commutative case the Hirota system~\eqref{eq:H-M} supplemented by the constraint~\eqref{eq:P-constr} coincides, up to simple change of the independent discrete variables, with the "multidimensional discrete-time Toda lattice" obtained in~\cite{AptekarevDerevyaginMikiVanAssche} within the theory of multiple orthogonal polynomials. Although the system of equations involves arbitrary number (but not less then two) of discrete variables, its initial-boundary data depend on functions of single variables~\cite{Doliwa-Siemaszko-HP}; for the Hermite--Pad\'{e} problem they are the $m$ sequences of the series coefficients. This shows that effectively the reduced system is two-dimensional. 
\end{Rem}
\begin{Rem}
The identity~\eqref{eq:P-constr-nc} satisfied by the solutions of the Hermite--Pad\'{e} problem was obtained under the non-degeneracy assumption. In the commutative case it was possible to derive~\cite{Doliwa-Siemaszko-HP} the constraint~\eqref{eq:P-constr} using standard properties of  determinants; see also~\cite{BakerGraves-Morris} for analogous direct proof of the identity~\eqref{eq:P-eq}. It seems possible that also in the non-commutative case equation~\eqref{eq:P-constr-nc} can be shown using more advanced quasideterminantal identities~\cite{KrobLeclerc} without that additional assumption. 
\end{Rem}
\subsection{Integrability of the constraint}
Let us leave the context of the Hermite--Pad\'{e} approximants, which give only special solutions to the non-Abelian Hirota--Miwa system. Their speciality is not only because of the constraint \eqref{eq:P-eq-nc-1}, but also because the solutions are defined for $n_k\geq -1$, $k=1,\dots , m$. We will study what impact on the solutions of the system has the additional constraint \eqref{eq:P-constr-nc} on the level of the linear problem~\eqref{eq:lin-d-adj-KP}. We will not be bounded to any special region of $\ZZ^m$.
\begin{Th}
	The compatibility condition of the linear system \eqref{eq:lin-d-adj-KP} with the additional equation \eqref{eq:P-constr-nc} is the non-Abelian Hirota--Miwa system~\eqref{eq:rho-rho}-\eqref{eq:rho-rho-rho} supplemented by the constraint
\begin{equation}
\label{eq:P-eq-nc}
\rho^1(n+e_1,x) [\rho^1(n)]^{-1} + \dots + \rho^m(n+e_m,x) [\rho^m(n)]^{-1} = 1 - F(|n|),
\end{equation}	
where $F\colon \ZZ \to \DD$ is a function of single variable $|n| = n_1 + \dots + n_m$.	
\end{Th}
\begin{proof}
We need to prove only the last part of the statement. Define 
\begin{equation*}
C(n) = 1 - \rho^1(n+e_1,x) [\rho^1(n)]^{-1} - \dots  - \rho^m(n+e_m,x) [\rho^m(n)]^{-1},
\end{equation*}	
and use directly the linear equations \eqref{eq:lin-d-adj-KP} and \eqref{eq:P-constr-nc} to obtain
\begin{equation*} 
\begin{split}
\bpsi(n+e_1,x) C(n) = \bpsi(n+e_1,x) - x \bpsi(n,x) + \bpsi(n+e_1 + e_2,x) \rho^2(n+e_1 + e_2)[\rho^2(n)]^{-1} + \\
 \bpsi(n+e_1 + e_3,x) \rho^3(n+e_1 + e_3)[\rho^3(n)]^{-1} + \dots  + \bpsi(n+e_1 + e_m,x) \rho^m(n+e_1 + e_m)[\rho^m(n)]^{-1}.
\end{split} 
\end{equation*}
Shift the above equation back in the $e_1$ direction and subtract the analogous equation with distinguished second variable
\begin{equation*} 
\begin{split}
\bpsi(n+e_2,x) C(n) = \bpsi(n+e_2,x) - x \bpsi(n,x) + \bpsi(n+e_1 + e_2,x) \rho^1(n+e_1 + e_2)[\rho^1(n)]^{-1} + \\
 \bpsi(n+e_2 + e_3,x) \rho^3(n+e_2 + e_3)[\rho^3(n)]^{-1} + \dots  + \bpsi(n+e_2 + e_m,x) \rho^m(n+e_2 + e_m)[\rho^m(n)]^{-1},
\end{split} 
\end{equation*}
shifted back in the $e_2$ direction, what gives
\begin{small}
\begin{gather*} 
\bpsi(n,x) [C(n-e_1) - C(n-e_2)]  =  - x [\bpsi(n-e_1,x) - \bpsi(n-e_2,x)] + \\
\bpsi(n+ e_2,x) \rho^2(n+ e_2)[\rho^2(n-e_1)]^{-1} 
- \bpsi(n+ e_1,x) \rho^1(n+ e_1)[\rho^1(n-e_2)]^{-1} + \\
\bpsi(n+ e_3,x) \rho^3(n + e_3)\left( [\rho^3(n-e_1)]^{-1} 
- \rho^3(n-e_2)]^{-1} \right) + \dots \\ \dots +
\bpsi(n+ e_m,x) \rho^m(n + e_m)\left( [\rho^m(n-e_1)]^{-1} 
- \rho^m(n-e_2)]^{-1} \right) .
\end{gather*} \end{small}
Using once again the linear equations \eqref{eq:lin-d-adj-KP} but also the non-Abelian Hirota--Miwa system~\eqref{eq:rho-rho}-\eqref{eq:rho-rho-rho} we obtain
\begin{gather*} 
\bpsi(n,x) [C(n-e_1) - C(n-e_2)] \rho^2(n-e_1) [\rho^2(n)]^{-1}= \\ - x \bpsi(n,x) +
\bpsi(n+e_1,x) \rho^1(n+e_1,x) [\rho^1(n)]^{-1} + \dots + \bpsi(n+e_m,x) \rho^m(n+e_m,x) [\rho^m(n)]^{-1} =0.
\end{gather*}
Doing that for any pair of indices we get finally
\begin{equation*}
C(n-e_i) = C(n-e_j), \qquad i,j = 1, \dots , m,
\end{equation*}	
what implies that the function $C(n)=F(|n|)$ depends only on the sum of all the variables.
\end{proof}
\begin{Rem}
	When values of $F$ commute with all potentials $\rho^i$ and in the allowed range of $k$ we have $F(k)\neq 1$, then one can remove the function from the non-linear equations (including the non-Abelian Hirota--Miwa part \eqref{eq:rho-rho}-\eqref{eq:rho-rho-rho}) by the transformation
	\begin{equation}
	\rho^j(n) = \tilde{\rho}^j(n)G(|n|),
	\end{equation}
where $G(k) = \prod_{i=1}^k (1 - F(k-i))$.
\end{Rem}

\section{Conclusion and open problems}
We presented the concept of the non-commutative Hermite--Pad\'{e} problem and we gave its solution in terms of quasideterminants. In our research we were guided by the corresponding results of the classical commutative problem and its relation to the Hirota system with the Paszkowski constraint. 

The theory of non-commutative Pad\'{e} approximants has close connection with the theory of matrix orthogonal polynomials, and the theory of Hermite--Pad\'{e} approximants has close connection with the theory of multiple orthogonal polynomials. Thus the results presented in the paper provide a link between (non-existing yet) theory of matrix multiple orthogonal polynomials with integrability. 

The non-commutative Hirota system provides a way~\cite{DoliwaKashaev} to construct certain maps satisfying Zamolodchikov's tetrahedron condition~\cite{Zamolodchikov}, which is a multidimensional generalization of the more familiar Yang--Baxter maps. The new class of solutions of the system, obtained in the present paper, can be used, in principle, to produce large family of the corresponding Zamolodchikov maps.

\providecommand{\bysame}{\leavevmode\hbox to3em{\hrulefill}\thinspace}


\begin{thebibliography}{10}

\bibitem{AblowitzBarYaacovFokas}	
M. J. Ablowitz, D. Bar Yaacov, A. S. Fokas, \emph{On the inverse scattering
	problem for the {Kadomtsev}--{Petviashvili} equation}, Stud. Appl.
Math. {\bf 69} (1983), 135--143.

\bibitem{AdlervanMoerbeke}
M. Adler, P. van Moerbeke, \emph{Generalized orthogonal polynomials, discrete KP and Riemann-Hilbert problems}, Comm. Math. Phys. \textbf{207} (1999) 589--620.

\bibitem{Alvarez-FernandezPrietoManas}
C. \'{A}lvarez-Fern\'{a}ndez, U. Fidalgo Prieto, M. Ma\~{n}as, \emph{Multiple orthogonal polynomials of mixed type:	Gauss--Borel factorization and the multi-component
	2D Toda hierarchy}, Adv. Math. \textbf{227} (2011) 1451--1525. 

\bibitem{Aptekarev}
A. I. Aptekarev, \emph{Multiple orthogonal polynomials}, J. Comput. Appl. Math. \textbf{99} (1998) 423--447.


\bibitem{AptekarevKuijlaars}
A. I. Aptekarev, A. B. J. Kuijlaars, \emph{Hermite--Pad\'{e} approximations and multiple orthogonal polynomial ensembles}, Russian Math. Surveys \textbf{66} (2011) 1133--1199.

\bibitem{AptekarevDerevyaginVanAssche}
A. I. Aptekarev, M. Derevyagin, W. Van Assche, \emph{Discrete integrable systems generated by Hermite--Pad\'{e} approximants}, Nonlinearity \textbf{29} (2016) 1487--1506.	


\bibitem{AptekarevDerevyaginMikiVanAssche}
A. I. Aptekarev, M. Derevyagin, H. Miki, W. Van Assche, \emph{Multidimensional Toda lattices: continuous and discrete time}, 
SIGMA \textbf{12} (2016) 054, 30 pp.	

\bibitem{AFACGAMM}
C. \'{A}lvarez-Fern\'{a}ndez, G. Ariznabarreta, J. C. Garc\'{i}a-Ardila, M. Ma\~{n}as, F. Marcell\'{a}n, \emph{
Christoffel Transformations for Matrix Orthogonal Polynomials in the Real Line and the non-Abelian 2D Toda Lattice}, 
Internat. Math. Res. Notices \textbf{2017} (2017) 1285--1341.

\bibitem{BakerGraves-Morris}
G. A. Baker, Jr., P. Graves-Morris, \emph{Pad\'{e} approximants}, 2nd edition, Cambridge University Press, Cambridge, 1996.

\bibitem{BBEIM}
E. D. Belokolos, A. I. Bobenko, V. Z. Enol’skii, A. R. Its,
V. B. Matveev, \emph{Algebro-geometric approach to
	nonlinear integrable equations}, Springer Series in Nonlinear Dynamics, Springer-Verlag, Berlin, 1994.


\bibitem{BleherKuijlaars}
P. M. Bleher, A. B. J. Kuijlaars, \emph{Random matrices with external source and multiply orthogonal polynomials}, Int. Math. Res. Not. \textbf{2004} (2004) 109--129.



\bibitem{Cohn}
P. M. Cohn, \emph{Skew fields. Theory of general division rings}, Cambridge University Press, 1995.


\bibitem{DKJM}
E. Date, M. Kashiwara, M. Jimbo, T. Miwa, \emph{Transformation groups for
	soliton equations}, [in:] Nonlinear integrable systems --- classical theory and
quantum theory, Proc. of RIMS Symposium, M. Jimbo and T. Miwa (eds.), World
Scientific, Singapore, 1983, 39--119.

\bibitem{DJM-II}
E. Date, M. Jimbo, T. Miwa, \emph{Method for generating discrete soliton
	equations. II}, J. Phys. Soc. Japan \textbf{51} (1982) 4125--31.


\bibitem{DD-DC-2}
J. Della Dora, C. Di Crescenzo, \emph{Approximants de Pad\'{e}--Hermite. 2\`{e}me partie: programmation}, Numer. Math. \textbf{43} (1984) 41--57.


\bibitem{DiFrancescoKedem}
P. Di Francesco, R. Kedem,
\emph{Non-commutative integrability, paths and
	quasi-determinants}, 
Adv. Math. \textbf{228} (2011) 97--152.



\bibitem{Dol-Des} 
A. Doliwa, \emph{Desargues maps and the Hirota--Miwa equation}, Proc. R. Soc. A
\textbf{466} (2010) 1177--1200.



\bibitem{Dol-AN} 
A. Doliwa, {\it The affine Weyl group symmetry of Desargues maps 
	and of the non-commutative Hirota--Miwa system}, Phys. Lett. A {\bf 375} (2011) 1219--1224.

\bibitem{Dol-GD}
A. Doliwa, \emph{Non-commutative lattice modified Gel'fand--Dikii systems}, J. Phys. A: Math. Theor. \textbf{46} (2013) 205202, 14~pp.


\bibitem{Doliwa-NCCF}
A. Doliwa, \emph{Non-commutative double-sided continued fractions}, J. Phys. A: Math. Theor. \textbf{53} (2020) 364001, 23 pp.

\bibitem{DoliwaKashaev}
A. Doliwa, R. M. Kashaev, \emph{Non-commutative bi-rational maps satisfying Zamolodchikov equation, and Desargues lattices}, J. Math. Phys. \textbf{61} (2020) 092704, 23pp.



\bibitem{DoliwaNoumi}
A. Doliwa, M. Noumi, \emph{The Coxeter relations and KP map in non-commuting symbols}, Lett. Math. Phys. \textbf{110} (2020) 2743--2762.


\bibitem{Doliwa-Siemaszko-W}
A. Doliwa, A. Siemaszko, \emph{Integrability and geometry of the Wynn recurrence}, \texttt{arXiv:2201.01749}


\bibitem{Doliwa-Siemaszko-HP}
A. Doliwa, A. Siemaszko, \emph{Hermite--Pad\'{e} approximation and integrability}, \texttt{arXiv:2201.06829}.

\bibitem{Draux-rev}
A. Draux, \emph{The Pad\'{e} approximants in a non-commutative algebra and their applications},[in:] H. Werner, H. J. B\"{u}nger (eds), Pad\'{e} Approximation and its Applications Bad Honnef 1983, Lecture Notes in Mathematics \textbf{1071} (1984) Springer, Berlin, Heidelberg.

\bibitem{Draux-OP-PA}
A. Draux, \emph{Formal orthogonal polynomials and Pade approximants in a non-commutative algebra}, [in:] P. A. Fuhrmann (ed.), Mathematical Theory of Networks and Systems, Lecture Notes in Control and Information Sciences \textbf{58} 1984, Springer, Berlin, Heidelberg.


\bibitem{Filipuk-VanAssche-Zhang}
G. Filipuk, W. Van Assche, L. Zhang, \emph{Ladder operators and differential equations for multiple orthogonal polynomials}, J. Phys. A: Math. Theor. \textbf{46} (2013) 205204.  


\bibitem{Quasideterminants-GR1}
I. Gelfand, V. Retakh, \emph{A Theory of noncommutative determinants
	and characteristic functions of graphs}, Funct. Anal. Appl. \textbf{26} (1992) 1--20. 

\bibitem{Quasideterminants}
I. Gelfand, S. Gelfand, V. Retakh, R. L. Wilson, \emph{Quasideterminants}, 
Adv. Math. \textbf{193} (2005) 56--141.

\bibitem{NCSF}
I.M. Gelfand, D. Krob, A. Lascoux, B. Leclerc, V.S. Retakh, J.-Y. Thibon, \emph{Non-commutative symmetric functions}, Adv. in Math. \textbf{112} (1995) 218--348.

\bibitem{GilsonNimmoOhta}
C. R. Gilson, J. J. C. Nimmo, Y. Ohta, \emph{Quasideterminant solutions of a non-Abelian
	Hirota-Miwa equation}, J. Phys. A \textbf{40} (2007)  12607.

\bibitem{Hermite}
C. Hermite, \emph{Sur la fonction exponentielle}, Oeuveres~III (1873) 150--181.

\bibitem{Hermite-P}
C. Hermite, \emph{Sur la g\'{e}n\'{e}ralisation des fractions continues alg\'{e}briques}, Oeuveres~IV (1893) 357--377.


\bibitem{Hirota}
R. Hirota,
\emph{Discrete analogue of a generalized Toda equation},
J. Phys. Soc. Jpn. \textbf{50} (1981) 3785--3791.


\bibitem{Konopelchenko-book}
B. G. Konopelchenko, \emph{Solitons in multidimensions}, World Scientific, Singapore, 1993.

\bibitem{Krichever}
I. M. Krichever, \emph{Two-dimensional periodic difference operators and algebraic geometry}, Dokl.
Akad. Nauk SSSR \textbf{285} (1985) 31--36.

\bibitem{KrobLeclerc}
D. Krob, B. Leclerc, \emph{Minor identities for quasi-determinants and quantum determinants}, 
Comm. Math. Phys. \textbf{169} (1995) 1-23.

\bibitem{Kuijlaars}
A. B. J. Kuijlaars, \emph{Multiple orthogonal polynomial ensembles. Recent trends in orthogonal polynomials and approximation theory}, Contemp. Math. \textbf{507} (2010) 155--176.


\bibitem{KNS-rev}
A. Kuniba, T. Nakanishi, J. Suzuki,
{\it $T$-systems and $Y$-systems in integrable systems},
J. Phys. A: Math. Theor. {\bf 44} (2011) 103001 (146pp).

\bibitem{LiNimmoTamizhmani}
C. X. Li, J. J. C. Nimmo, K. M. Tamizhmani,
\emph{On solutions to the non-Abelian Hirota--Miwa
equation and its continuum limits},
Proc. R. Soc. A  \textbf{465} (2090) 1441--1451.

\bibitem{Shi-HaoLi}
S.-H. Li, \emph{Matrix orthogonal polynomials, non-abelian Toda lattice and B\"acklund transformations}, \texttt{arXiv:2109.00671}. 


\bibitem{LopezLagomasinoMedinaPeraltaSzmigielski}
G. L\'{o}pez Logomasino, S. Medina Peralta, J. Szmigielski, \emph{Mixed type Hermite--Pad\'{e} approximation inspired by the Degasperis--Procesi equation}, Adv. Math. \textbf{349} (2019) 813--838.


\bibitem{Mahler-P}
K. Mahler, \emph{Perfect systems}, Compositio Math. \textbf{19} (1968) 95--166.


\bibitem{ManoTsuda}
T. Mano, T. Tsuda, \emph{Hermite--Pad\'{e} approximation, isomonodromic deformation and hypergeometric integral}, Math. Zeitschrift \textbf{285} (2016) 397--431.


\bibitem{Matveev}
V. B. Matveev, M. A. Salle, \emph{Darboux transformations and solitons}, Springer Series in Nonlinear Dynamics, Springer-Verlag, Berlin, 1991.

\bibitem{Miranian}
L. Miranian, \emph{Matrix-valued orthogonal polynomials on the real line: some extensions of the classical theory}, J. Phys. A: Math. Gen. \textbf{38} (2005) 5731.

\bibitem{Miwa} 
T.~Miwa, \emph{On Hirota's difference equations}, 
Proc. Japan Acad. \textbf{58} (1982) 9--12.






\bibitem{Nimmo-KP}
J. J. C. Nimmo, \emph{Darboux transformations and the discrete KP equation}, 
J. Phys. A: Math. Gen. \textbf{30} (1997) 8693--8704. 



\bibitem{Nimmo-NCKP}
J. J. C. Nimmo, \emph{On a non-Abelian Hirota-Miwa equation}, 
J. Phys. A: Math. Gen. \textbf{39} (2006) 5053--5065. 



\bibitem{Paszkowski}
S. Paszkowski, \emph{Recurrence relations in Pad\'{e}--Hermite approximation}, J. Comput.
Appl. Math. \textbf{19}  (1987) 99--107.


\bibitem{Shiota}
T. Shiota, \emph{Characterization of Jacobian varieties in terms of soliton equations}, Invent. Math. \textbf{83} (1986) 333--382.

\bibitem{SinapVanAssche}
A. Sinap, W. Van Assche,
\emph{Orthogonal matrix polynomials and applications}, 
J. Comput. Appl, Math. \textbf{66} (1996) 27--52.

 

\bibitem{VanAsche}
W. Van Assche, \emph{Pad\'{e} and Hermite--Pad\'{e} approximation and orthogonality}, Surv. Approx. Theory \textbf{2} (2006) 61--91.

\bibitem{Zabrodin}
A. V. Zabrodin, \emph{Hirota’s difference equations}, 
Theor. Math. Phys. \textbf{113} (1997) 1347--1392.



\bibitem{Zamolodchikov}
A. B. Zamolodchikov, \emph{Tetrahedron equations and the relativistic $S$-matrix of 
	straight-strings in $2+1$ dimensions}, Commun. Math. Phys. \textbf{79} (1981) 489--505.

\end{thebibliography}
\end{document}